\documentclass[preprintnumbers,showpacs,amsmath,amssymb,floatfix,prd,onecolumn,superscriptaddress,nofootinbib]{revtex4}
\usepackage{graphicx}
\usepackage{epsfig}
\usepackage{bm}
\usepackage{amsfonts}
\usepackage{epstopdf}
\usepackage{bm}
\usepackage{amsthm}

\newtheorem{thm}{Lemma}[section]

\begin{document}

\title{ Take up the challenge for a single field inflation after BICEP2}

\author{Chao-Jun Feng}
\email{fengcj@shnu.edu.cn} 
\affiliation{Shanghai United Center for Astrophysics (SUCA), \\ Shanghai Normal University,
    100 Guilin Road, Shanghai 200234, P.R.China}
    \affiliation{State Key Laboratory of Theoretical Physics, \\Institute of Theoretical Physics, Chinese Academy of Sciences, Beijing 100190, P.R.China}

\author{Xin-Zhou Li}
\email{kychz@shnu.edu.cn} \affiliation{Shanghai United Center for Astrophysics (SUCA),  \\ Shanghai Normal University,
    100 Guilin Road, Shanghai 200234, P.R.China}

\begin{abstract}
The measurement of the tensor-to-scalar ratio $r$ shows a very powerful  constraint to theoretical inflation models through the detection of B-mode. In this paper, we propose a single inflation model with infinity pows series potential called the Infinity Power Series (IPS) inflation  model, which is well consistent with  latest observations from \textit{Planck} and BICEP2. Furthermore,  it is found that in the IPS model, the absolute value the running of the spectral index  increases while the tensor-to-scalar ratio becomes large, namely both large $r \approx 0.20$ and $|n_s'|\gtrsim \mathcal{O}(10^{-2})$ are realized in the IPS model. In the meanwhile,  the number of e-folds is also big enough  $N\approx 35\sim45$ to solve the problems in Big Bang cosmology.
\end{abstract}

\maketitle


\section{Introduction}

The B-mode polarization that can only be generated by the tensor modes has been observed by BICEP2 group \cite{Ade:2014xna}. This is the first evidence for the primordial gravitational wave.  According to the results of BICEP2, the tensor-to-scalar ratio is constrained to $r = 0.20_{-0.05}^{+0.07}$ at  $68\%$ level for the lensed-$\Lambda$CDM model, with $r=0$ disfavoured at $7.0\sigma$ level.  With the help of this new constant on $r$ as well as that on the spectral index, some inflation models with prediction of negligible tensor perturbation have been excluded, such at the small-field inflation models.

There is a tension between the observational results from BICEP2 and the report of \textit{Planck}.
According to the recent  works from \textit{Planck} \cite{Ade:2013uln}, it gives $n_s=0.9600\pm0.0071$ and $r_{0.002}<0.11$ at $95\%$CL by combination of \textit{Planck}+WP+highL data.  However, when the running parameter defined as $n'_s \equiv dn_s/d\ln k$ is included in the data fitting, it gives $n_s=0.95700\pm0.0075$, $r_{0.002}<0.26$ and  $n'_s = -0.022_{-0.021}^{+0.020}$ at $95\%$CL by the same data.  In a word, to give a consistent constraint on $r$ for the combination of  \textit{Planck}+WP+highL  data and the BICEP2 data, we require a running of the spectral index $n'_s\lesssim -0.001$ at $95\%$CL.

On the other hand, inflation is the most economic way to solve the puzzles like the horizon problem in the Hot Big Bang cosmology. It also generates a small quantum fluctuation that could be considered as the seed of the large-scale structure in the later time. The detection of B-mode from CMB by the BICEP2 group also indicates a strong evidence of inflation \cite{Guth:1980zm,Linde:1981mu,Albrecht:1982wi}. In a simplest slow-roll inflation model, the early universe was driven by a single scalar field $\phi$ called the inflaton with a very flat potential $V(\phi)$. For example, in the so-called chaotic inflation \cite{Linde:1983gd}, we have $V\sim \phi^n$, while in the so-called natural inflation \cite{Freese:1990rb}, we have $V\sim V_0(\cos{\phi/f}+1)$. In these  kind of single inflation models, the spectral index $n_s$ for the scalar perturbation deviates from the Harrison-Zel'dovich value of unity in order of $n_s-1\sim \mathcal{O}(10^{-2})$, and its running in order of $n'_s\sim \mathcal{O}(10^{-4})$. Thus the explanation of large $r$ and $|n'_s|$ is a challenge to single scalar field, see ref.\cite{Gong:2014cqa} for much detail discussions on this challenge.

Also, in ref.\cite{Chung:2003iu} the authors pointed out that it is difficult to achieve in common realisations of slow-roll inflation as it requires a large third field derivative while maintaining small first and second derivatives. Usually, to give a stronger running of the spectral index, a small number of e-folds $N$ is needed, e.g. $N\approx23$ contributes a factor $1/23^2\sim0.002$ to $n'_s$ \cite{Chung:2003iu}, see also ref.\cite{Easther:2006tv} for the same conclusions. However, these discussion are all based on the assumption that the inflaton potential is expanded in the Taylor series of the inflaton field with finite truncation.  In this paper, in order to take up this challenge for the single inflation models, we propose a infinity power series potential for the inflaton, see Eq.(\ref{equ:potential}). We find that the absolute value ??the running of the spectral index is increasing while the tensor-to-scalar ratio becomes large. In the meanwhile,  the number of e-folds is also big enough to solve the problems in Big Bang cosmology, say $N\approx 35\sim45$. And also it is proved that this kind of infinity power series potential is convergent when the parameters are located at some physical regions in the parameter space. In  next section, we calculate the slow-roll parameters for this model and one can see that it indeed predict a large tensor-to-scalar ratio and the running of the spectral index, and we also confront this model with latest observations from \textit{Planck} and BICEP2. In last section, we will give some conclusions and discussions.  The convergence proof and some coefficients of the potential are presented in the Appendix.\ref{app:pro} and \ref{app:cof}. For related works to realize  a large running spectral index, see refs.\cite{Feng:2003mk,Kobayashi:2010pz}.

\section{Single field inflation with power series potential }

\subsection{Potential of Infinity Power Series (IPS) inflation  model.}
In this paper, we propose a single field inflation with the following power series potential in the units of $M_{pl}^{-2} = 8\pi G =1$:
\begin{equation}\label{equ:potential}
	V(\phi) = \alpha \sum_{n=1}^{\infty} a_n \phi^{2n} \,, 
\end{equation}
where $\alpha$ is a proportional constant, and we call it the Infinity Power Series (IPS) inflation  model. Here $a_n$ are nothing but some constant coefficients. Since $a_1$ could be any finite values, we absorb it into the parameter $\alpha$. While for $n\geq2$, it satisfies the following recurrence  relation:
\begin{equation}\label{equ:reca}
	a_n = -\frac{1 }{8Aa_1n(2n-1)(n-1)}  \sum_{k=1}^{n-1} a_{n-k}\bigg\{  k\bigg[ (6-B)k-2 -  (2-B) n \bigg] a_k + 8A  (k+1) (n-k)(2n-2k-1)(n-k-1)a_{k+1}  \bigg\}\,,
\end{equation}
where $A, B$ are two parameters in this model, which  can be determined by observations.   By using the above recurrence  relation (\ref{equ:reca}), one could get the value of $a_n$ for a given $n$. We also present the values of $a_n$ for some leading terms of the potential in the App.\ref{app:cof}, in which it shows that the value of $a_n$ is decreasing with $n$ increased for a given parameter $A$.  

Such a kind of power series potential may be derived from some supergravity (SUGRA) theories :
\begin{equation}
	V = e^K \bigg[ K^{i\bar j} (D_i W)(D_{\bar j}\bar W ) - 3|W|^2\bigg]\,, \quad D_iW = \partial_i W + (\partial_i K) W\,,
\end{equation}
with the  the following K\"ahler potential for example:
\begin{equation}\label{equ:khler}
	K = g_1(\varphi + \varphi^\dag) +  \frac{1}{2}(\varphi + \varphi^\dag)^2 + |X|^2 - c_X|X|^4 + \cdots \,, 
\end{equation}
which respects the shift symmetry to ensure the flatness of the potential beyond the Planck scale along the imaginary components of $\varphi$, $\phi\equiv \sqrt{2}\text{Im}\varphi$. Here $\phi$ is identified to be the inflaton and $X$ is another chiral superfield. The superpotential is assumed to be of the general form
\begin{equation}\label{equ:super}
	W = X\sum h_n\phi^{n} + W_0 \,.
\end{equation}
see refs.\cite{Nakayama:2013txa} for details on the SUGRA theory. The authors in refs.\cite{Nakayama:2013txa, Kobayashi:2014jga} have considered so-called the polynomial inflation models with limited terms in the potential.

Notice that many forms of potential can be  expanded as an infinity power series, e.g. an exponential potential $V\sim e^\phi$ can be  written as $V \sim \sum \phi^n/n!$. Even a general potential can be expanded as Taylor series. However, in the Taylor series expansion, one can not control the infinity number of coefficients. The only way to deal with them seems to be firstly expanding the potential to only a few orders near the vicinity of the inflaton field, and then reconstructing the potential  from observations, see refs.\cite{Ma:2014vua, Choudhury:2014kma} for examples. In the IPS inflation model, there are only three parameters in the potential (\ref{equ:potential}), namely, $\alpha, A$ and $B$.  Furthermore, in the limit of $|A|\rightarrow\infty$ or $B\rightarrow 0$, one can recover the quadratic chaotic inflation model,  since in this case $a_n \rightarrow 0$ for $n\geq2$ , namely $V=\alpha \phi^2$. 

\begin{figure}[h]
\begin{center}
\includegraphics[width=0.45\textwidth,angle=0]{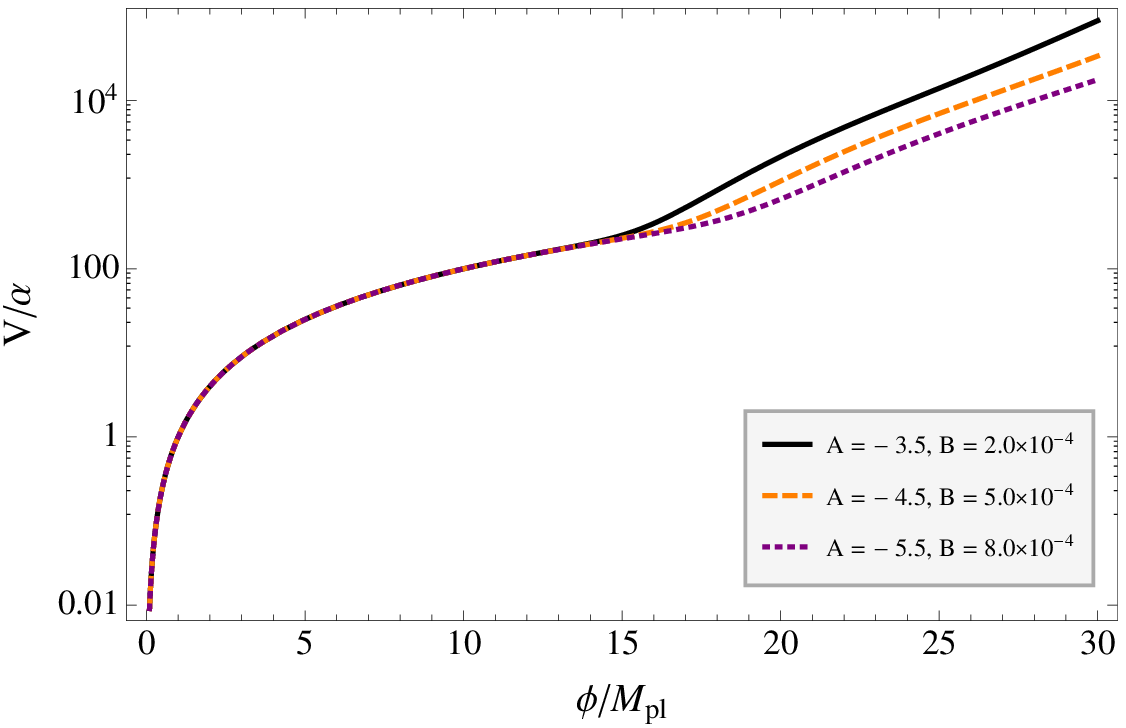}
\quad
\includegraphics[width=0.45\textwidth,angle=0]{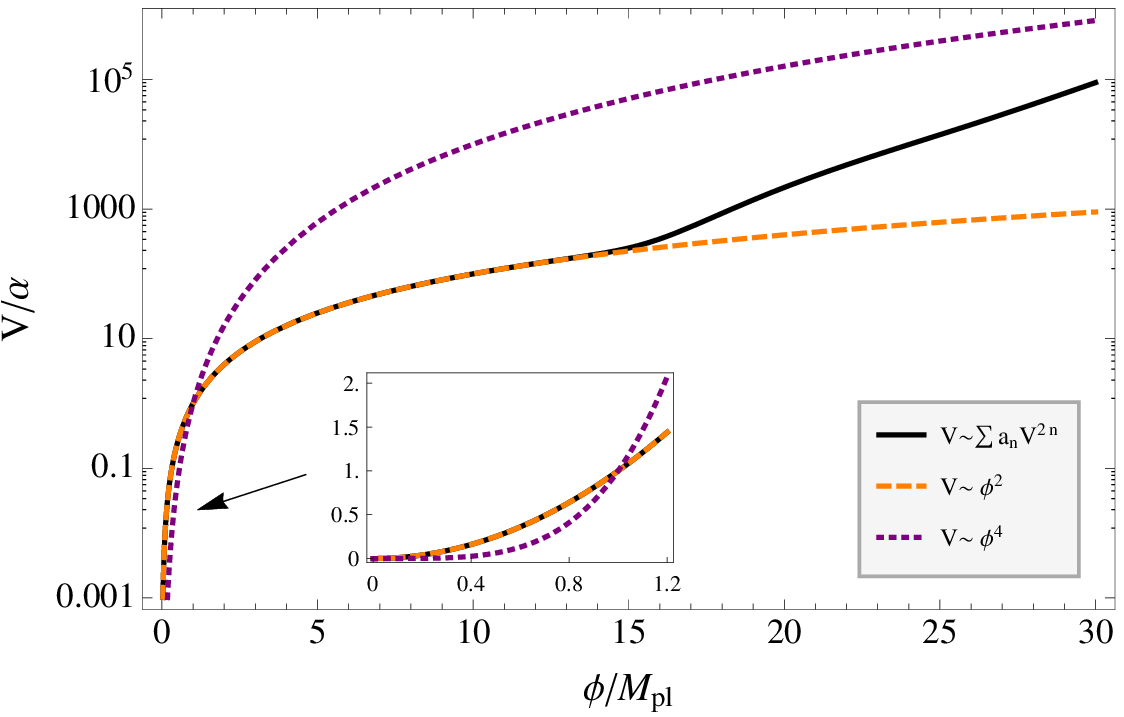}
\caption{\label{fig:pot} Left: The potential $V$ in Eq.(\ref{equ:potential}) with different parameters $A, B$. Right:  Compare to the chaotic inflation models with potentials $V\sim \phi^2, \phi^4$.     }
\end{center}
\end{figure}

From Fig.\ref{fig:pot}, one can see that the potential  is increasing with $\phi$ increased, and it gets a little bit more steep when the parameter $A$ becomes small. Compare this potential from Eq.(\ref{equ:potential}) with that in the chaotic inflation, e.g. $\phi^2$ or $\phi^4$, we find that it could be almost identical to the quadratic one when $\phi$ is small, while it becomes larger than $\phi^2$ when $\phi$ goes large, but it is smaller than the quartic potential. Furthermore,  one can see that this kind of infinity power series potential is convergent in the physical regions of parameter spaces, which is required for an successful inflation, see Appendix.\ref{app:pro}.

\subsection{Slow-roll parameters}

The slow-roll parameters are defined as  usual
\begin{equation}\label{equ:slowroll}
	\epsilon = \frac{M_{pl}^2}{2} \left( \frac{V' }{V} \right)^2 \,, \quad
	\eta = M_{pl}^2 \frac{V''}{V} \,, \quad
	\xi = M_{pl}^4 \frac{V'V'''}{V^2} \,,
\end{equation}
where  the prime denotes the derivatives with respect to $\phi$. Here we have recovered the unit $M_{pl}$ to indicate these parameters are all dimensionless.  By  substituting  the potential (\ref{equ:potential}) into these definitions, we obtain
\begin{equation}\label{equ:sl1}
	\epsilon = \frac{1}{2} \phi^{-2} \sum_{n=0}^\infty b_n \phi^{2n} \,,\quad 
	\eta = \phi^{-2} \sum_{n=0}^{\infty}  c_n \phi^{2n} \,,\quad
	\xi = \phi^{-2} \sum_{n=0}^{\infty}  d_n \phi^{2n}
\end{equation}
where $b_0 = 4, c_0 = 2, d_0 = -B/A $ and $b_n, c_n, d_n$ (for $n\geq 1$) satisfy the following recurrence  relations
\begin{eqnarray}
	b_n &=&  4(n+1)   a_{n+1}  + \sum_{k=1}^{n} \bigg[ 4(k+1)(n-k+1)a_{k+1}a_{n-k+1}  -   \sum_{m=1}^{k+1} a_m a_{k+2-m}  b_{n-k} \bigg] \,, \label{equ:recb}\\
	c_n &=&   2(n+1)(2n+1)  a_{n+1} - \sum_{k=1}^{n} a_{k+1} c_{n-k}    \label{equ:recc} \,,\\
\nonumber
	d_n &=&  8(n+2)(2n+3)(n+1)a_{n+2}  \\
	&+& \sum_{k=1}^{n}  \left [  8(k+1) (n-k+2)(2n-2k+3)(n-k+1) a_{k+1}a_{n-k+2}  - \sum_{m=1}^{k+1} a_ma_{k+2-m}  d_{n-k}  \right]\,.\label{equ:recd} 
\end{eqnarray}
Furthermore, by using the recurrence  relation of $a_{n}$ in Eq.(\ref{equ:reca}),  we find  the following relation between the coefficients of $b_n, c_n$ and $d_n$, see Appendix.\ref{app:cof} for details:
\begin{equation}\label{equ:bcd}
	(2-B)b_n - 4c_n = 4A~d_n \,,
\end{equation}
for $n=0, 1, \cdots$. Thus we obtain the relation between the slow-roll parameters 
\begin{equation}\label{equ:xia}
	\xi = \frac{1}{2A}\bigg[ (2-B)\epsilon - 2\eta\bigg]\,.
\end{equation}
Furthermore, the number of e-folds before the end of inflation is given by
\begin{equation}\label{equ:efold}
	N = \int_t^{t_{\text{end}}} H dt = -\int_{\phi}^{\phi_{\text{end}}} \frac{V}{V'} d\phi  = \frac{\phi^2}{4}\sum_{k=0}^{\infty} \frac{f_n}{n+1} (\phi^{2n} -\phi^{2n}_{\rm end} )\,,
\end{equation}
 with $f_0 =1$ and 
 \begin{equation}\label{equ:recf}
	f_n=  a_{n+1} - \sum_{k=1}^{n}(k+1)a_{k+1} f_{n-k} \,,
\end{equation}
for $n\geq1$. Here $\phi_{\rm end}$ denotes the value of the inflaton at the end of inflation, which is determined by $\epsilon\approx1$, or  by the following equation:
\begin{equation}\label{equ:phiend}
	2\phi^2_{\rm end} \approx \sum_{n=0}^\infty b_n \phi_{\rm end}^{2n} \,,
\end{equation}
which is around the Planck scale $\phi_{\rm end} /_{M_{pl}}\sim 1 $, and one can also numerical solve the above equation for $\phi_{\rm end}$.

\subsection{Power spectral index and tensor modes}
The power spectrum of the scalar perturbation is given by 
\begin{equation}\label{equ:ps}
	\mathcal{P}_s = A_s \left( \frac{k}{k_*}\right)^{n_s-1+ n_s'\ln(k/k_*)/2} \,,
\end{equation}
where $A_s$ is its amplitude, $n_s$ is called the spectral index, and $n_s' =  dn_s/d\ln k$ is the running of index. In a single field inflation model, the spectral indices and its running could be given in terms of the slow-roll parameters:
\begin{eqnarray}
	n_s - 1 &\approx&  2\eta - 6\epsilon \,,  \\
	n_s'  &\approx& 16\epsilon\eta-24\epsilon^2 -2\xi = 16\epsilon\eta-24\epsilon^2 -A^{-1}\bigg[ (2-B)\epsilon - 2\eta\bigg]\,,
\end{eqnarray}
where we have used Eq.(\ref{equ:xia}).  While the tensor-to-scalar ratio $r$ is simply given by $r\approx 16\epsilon$.  Here, one can see the parameter $A, B$ is very important to get a large running of spectral index.  From Fig.\ref{fig:nrns}, the  running of index is in order of $n'_s\sim \mathcal{O}(\epsilon) \sim  \mathcal{O}(10^{-2}) $, while the tensor ratio is $r \approx 0.20$. Therefore, both large $r$ and $|n_s'|$ are realized in the IPS model. The corresponding number of e-folds is $N\approx35\sim45$ for $n_s \approx 0.96$, which seems enough big to solve the problems in the Big Bang cosmology.

\begin{figure}[h]
\begin{center}
\includegraphics[width=0.43\textwidth,angle=0]{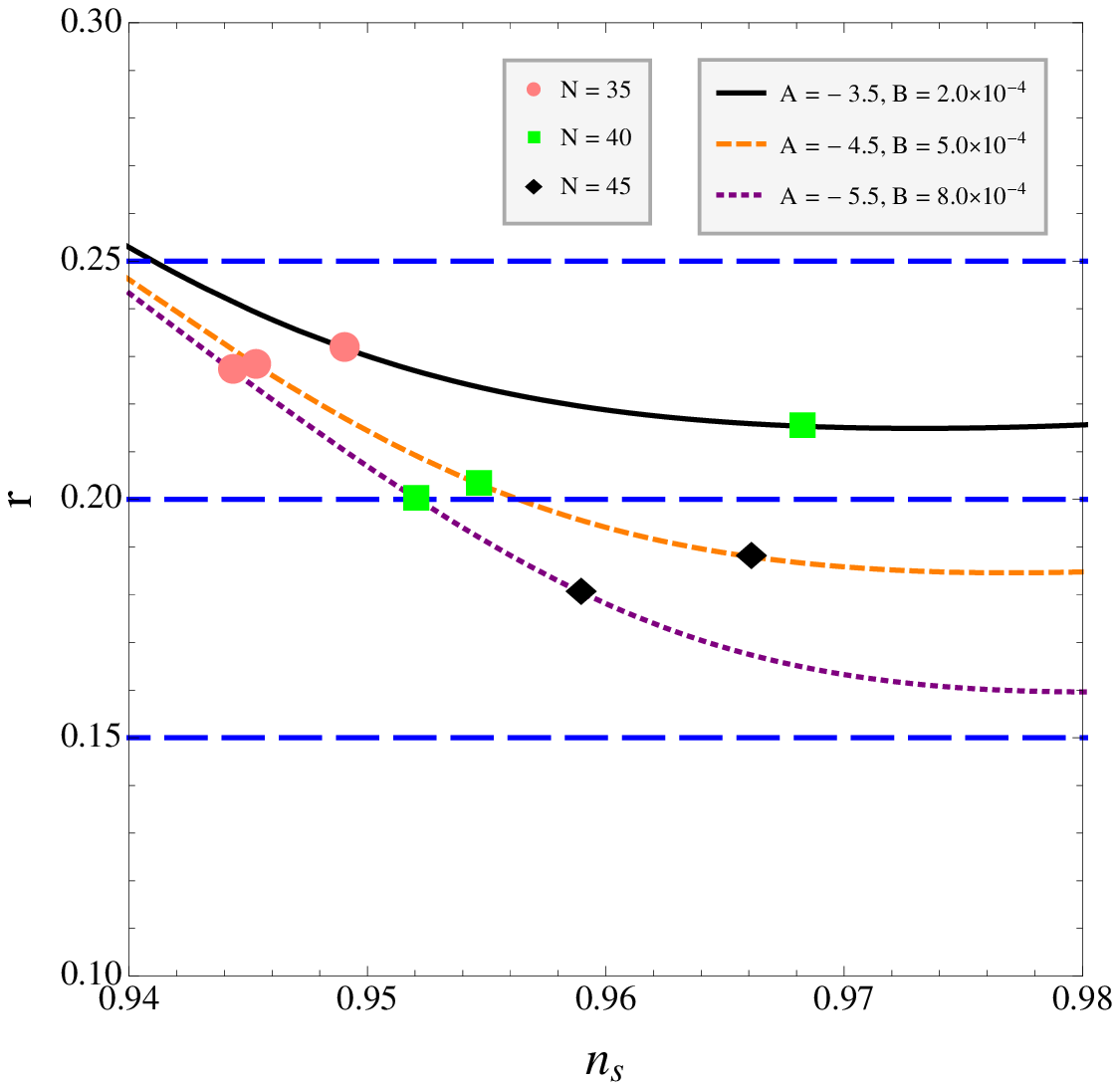}
\quad
\includegraphics[width=0.45\textwidth,angle=0]{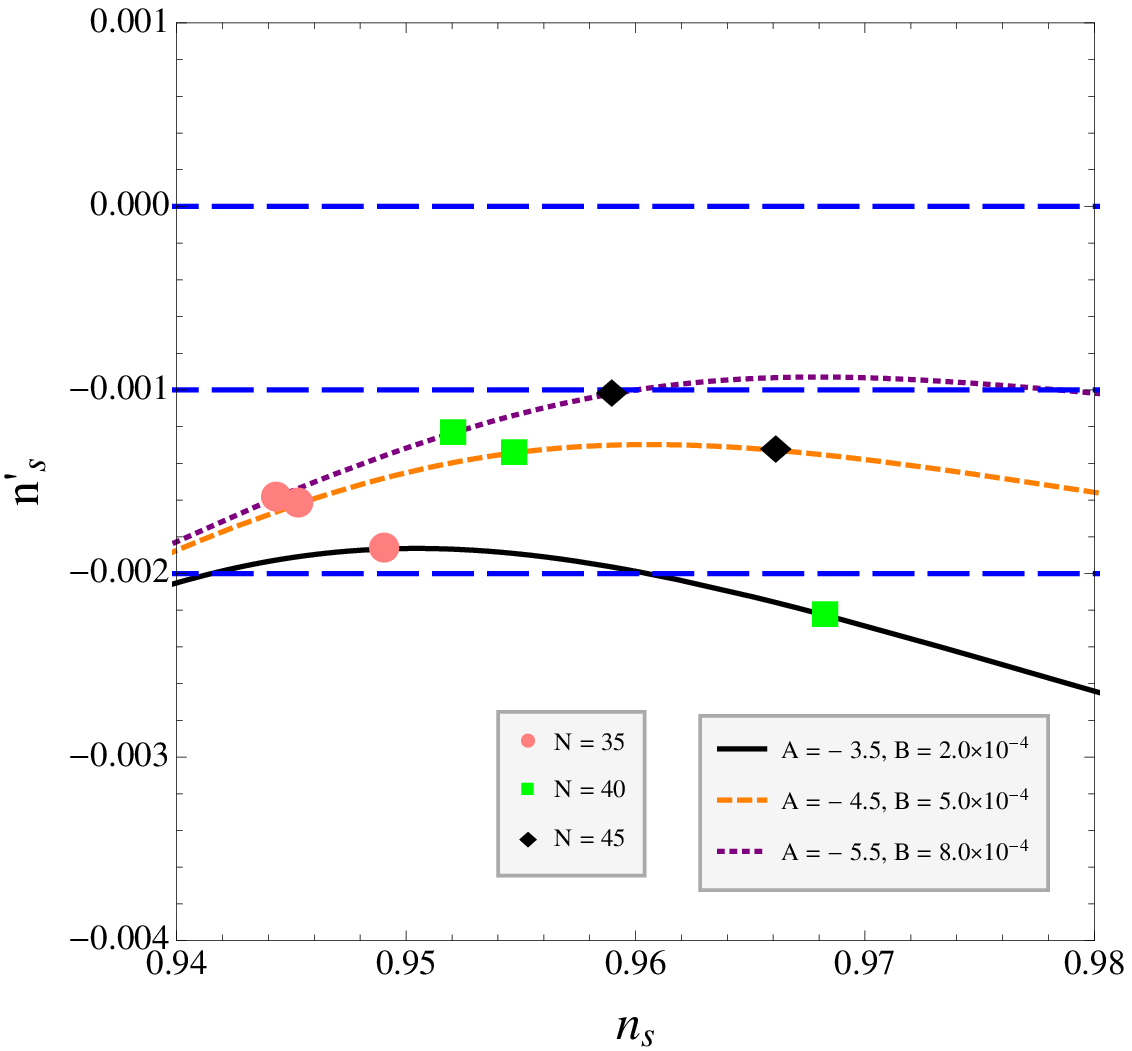}
\caption{\label{fig:nrns} Left: The tensor-to-scalar ratio $r \,\,\,  v.s. $ the spectral index $n_s$ with parameters $A, B$. Right:  The running of the spectral index $n'_s \,\,\, v.s.$ the spectral index $n_s$ with the same parameters.}
\end{center}
\end{figure}

The latest analysis of the data including the $Planck$ CMB temperature data, the WMAP large scale polarization data  (WP) , CMB data extending the \textit{Planck} data to higher-$l$,  the \textit{Planck} lensing power spectrum, and BAO data  gives the constraint on the index $n_s$ of the scalar power spectrum\cite{Ade:2013uln}: $ 0.9583\pm0.0081$(\textit{Planck}+ WP), $ 0.9633\pm0.0072$ (\textit{Planck}+WP+lensing), $ 0.9570\pm0.0075$ (\textit{Planck}+WP+highL), $ 0.9607\pm0.0063$ (\textit{Planck} +WP+BAO).  It also gives an upper bound on $r\lesssim 0.25$.  The BICEP2 experiment constraints the tensor-scalar-ratio as: $r=0.20^{+0.07}_{-0.05}$ in ref.\cite{Ade:2014xna}. They are also other groups have reported their constrain results on the ratio: $r=0.23^{+0.05}_{-0.09}$ in ref.\cite{Cheng:2014ota} by adopting the Background Imaging of Cosmic Extragalactic Polarization (B2), \textit{Planck} and WP data sets;  $r=0.20^{+0.04}_{-0.05}$  in ref.\cite{Cheng:2014cja} combined with the Supernova Legacy Survey (SNLS); $r=0.199^{+0.037}_{-0.044}$ in ref.\cite{Li:2014cka}  by adopting the \textit{Planck}, supernova \textit{Union2.1} compilation, BAO and BICEP2 data sets; and also $r=0.20^{+0.04}_{-0.06}$ in ref.\cite{Wu:2014qxa} with other BAO data sets. This  B-mode signal can not be mimicked by topological defects\cite{Lizarraga:2014eaa}, and also can not be explained in large extra- dimension models \cite{Ho:2014xza}. The most likely origin of this signal is from the tensor perturbations or the  gravitational wave polarizations during inflation.  And also, the cosmological constant seems important during inflation to predict a large tensor mode\cite{Feng:2014naa}.  

By using these results,  we obtain the constraints on the parameters $A$ and $B$, see Fig.\ref{fig:fit}. And the value of inflaton field during inflation is typically $\phi /M_{pl}\approx \mathcal{O}(10)$, which corresponds to a large field inflation.  By using the amplitude value of the power spectrum from \textit{Planck}, $A_s (k=0.002 Mpc^{-1}) = 2.215\times 10^{-19}$ \cite{Ade:2013zuv}, we can also estimate the value of the $\alpha$ parameter:
\begin{equation}\label{equ:estalpha}
	\alpha \approx \frac{3\pi^2 }{2 }\left (\sum_{n=1}^\infty a_n\phi^{2n}\right)^{-1} r A_s  \approx 4.37\times 10^{-11} \,,
\end{equation}
with the best fitting values of $A$ and $B$.

\begin{figure}[h]
\begin{center}
\includegraphics[width=0.43\textwidth,angle=0]{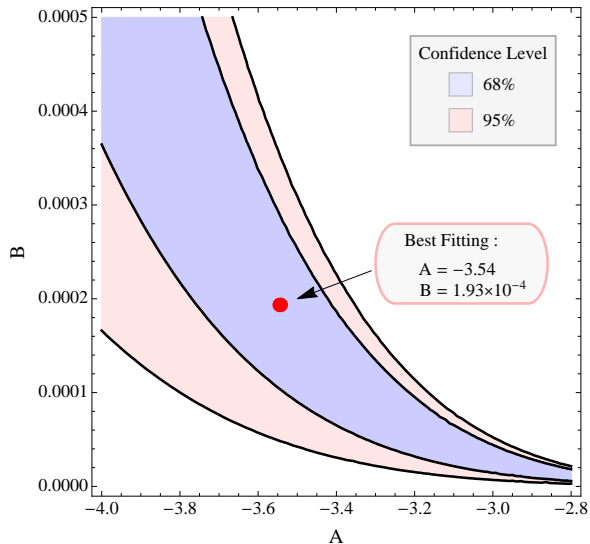}
\caption{\label{fig:fit}  Constraints on the values of parameters $A$ and $B$ with contours from $1\sigma$ to $2\sigma$ confidence levels.}
\end{center}
\end{figure}

\section{Conclusion}

In this paper, we have proposed the Infinity Power Series (IPS) inflation  model with the inflaton potential like Eq.(\ref{equ:potential}).  By confronting this model with latest observations from \textit{Planck} and BICEP2, we constrain the parameters $A$ and $B$. It is found that the absolute value of the running of the spectral index increases with the tensor-to-scalar ratio becomes large,  while the number of e-folds is around $35\sim45$, which seems enough to solve the problems in the Big Bang cosmology, see Fig.\ref{fig:nrns}.  In a word, the IPS model could predict both a large tensor ratio $r\approx0.2$ and a large running of spectral index $|n'_s | \gtrsim \mathcal{O}(10^{-2})$ as well as a large enough number of e-folds, and it also predict the spectral index around $n_s\approx 0.96$. 

The convergence test has been performed to the inflaton potential in some physical regions of the parameter space. There could be other regions that the potential could be convergent. Furthermore, this kind of potentials could also come from some SUGRA theories, and it deserves further study.

\acknowledgments

This work is supported by National Science Foundation of China grant Nos.~11105091 and~11047138, ``Chen Guang" project supported by Shanghai Municipal Education Commission and Shanghai Education Development Foundation Grant No. 12CG51, National Education Foundation of China grant  No.~2009312711004, Shanghai Natural Science Foundation, China grant No.~10ZR1422000, Key Project of Chinese Ministry of Education grant, No.~211059,  and  Shanghai Special Education Foundation, No.~ssd10004, and the Program of Shanghai Normal University (DXL124).

\appendix 
\section{Convergence of the inflaton potential}\label{app:pro}
In this section, we will prove that the  inflaton potential with infinite power series like Eq.(\ref{equ:potential}) is convergent absolutely in some regions  of the parameter space $\{A, B\}$. And within this convergent region, one could get an successful inflation with consistent predictions as described in this paper.  First of all, we will prove that the following theorem
\begin{thm}
By given $(-A)  \approx  10^p, $and $(-B/A) \approx 10^{-2q} $ with $ 0<p<2q$, then the following equation
\begin{equation}\label{equ:prov1}
	0 \leq a_n < \frac{a_1B}{8n(2n-1)(n-1)}  \left(\frac{B}{-A}\right)^{(n-2)/\beta} \,,
\end{equation}
 is always valid  for $n\geq 2$ as long as $\beta \gtrsim 2q/(q+1)$.  
\end{thm}

\begin{proof}
By using the recurrence relation (\ref{equ:reca}), we get 
\begin{equation}\label{equ:a23}
	a_2 = -\frac{B}{48A} a_1\,, \quad a_3 = \frac{B (B+3)}{5760 A^2} a_1\,,
\end{equation}
and both of them satisfy Eq.(\ref{equ:prov1})  when $A<-1, B\geq0$ and $-B/A\ll1$
\begin{equation}
	a_2 < \frac{B}{48} a_1\,, \quad a_3 < \frac{a_1}{240} B \left(\frac{B}{-A}\right)^{1/\beta}
\end{equation}
If we assume Eq.(\ref{equ:prov1}) is satisfied for $ 3<m\leq n$, then from Eq.(\ref{equ:reca}) we obtain
\begin{eqnarray}
\nonumber
	 a_{n+1} &=&-\frac{1 }{8Aa_1n(n+1)(2n+1)} \bigg\{  \sum_{k=1}^{n-2}  (n-k)\bigg[ (6-B)(n-k)-2 -  (2-B) (n+1) \bigg] a_{k+1} a_{n-k}\\
\nonumber
	&+& 8A\sum_{k=1}^{n-1}   (k+1) (n-k+1)(2n-2k+1)(n-k)a_{k+1}  a_{n-k+1} + 2 \left(B n+2 n^2-3 n+1\right) a_n a_{1}\bigg\}\\
	&<&\frac{ a_1B }{8(n+1)(2n+1)n} \left(\frac{B}{-A}\right)^{(n-1)/\beta}  F(A,B,n) \,,
\end{eqnarray}
where the function $F$ is given by
\begin{eqnarray}
\nonumber
 && F(A,B,n) =  \frac{1 }{8}\bigg\{
	\sum_{k=1}^{n-2} \left(\frac{B}{-A}\right)^{1-2/\beta} \bigg[  \frac{B }{8k(2k+1)(2n-2k-1)(n-k-1)}  
	-   \frac{ 1}{4(k+1)(2k+1)(2n-2k-1)(n-k-1)} \\
\nonumber
	&+&   \frac{1}{2k(k+1)(2k+1)(2n-2k-1)} \bigg] 
	+ \left(\frac{B}{-A}\right)^{1-1/\beta}  A\sum_{k=1}^{n-1}  \bigg[ \frac{1}{k}  - \frac{2}{2k+1}  \bigg]
	- \frac{1}{A} \left(\frac{B}{-A}\right)^{-1/\beta} \frac{2}{n} \frac{ \left(2 n^2-3 n+1 + Bn\right)}{(2n^2-3n+1)}  \bigg\} \\
	&\lesssim&  \frac{2n-5}{2 n-3}10^{-2\left[(q+1)\beta-2q\right]/\beta} 
	- \frac{1}{8}10^{-(2q-p)} \bigg(\psi ^{(0)}(n)-\psi ^{(0)}\left(n+1/2\right)+\gamma +\psi ^{(0)}\left(3/2\right)\bigg)< 1\,.
\end{eqnarray}
Thus,  $a_{n+1}$ also satisfies Eq.(\ref{equ:prov1}). 

\end{proof}

By using this lemma, we perform the Cauchy convergence test on the power series of the potential in Eq.(\ref{equ:potential})
\begin{equation}
	\lim_{n\rightarrow \infty}\left( a_n\phi^{2n}\right)^{1/n} < \lim_{n\rightarrow \infty} n^{-3/n} \left(\frac{B}{-A}\right)^{(n-2)/\beta n}\phi^2 \approx 10^{-(q+1)} \frac{\phi^2 }{M_{pl}^2}< 1\,, \quad \text{for} \quad \frac{\phi^2 }{M_{pl}^2}<10^{q+1}\,.
\end{equation}
In this paper, we obtain the parameters $(-B/A) \approx 10^{-5}\sim10^{-6}$ ( or $q\approx2\sim3$),  and the value of the inflation field during inflation is about $\phi/M_{pl} \lesssim \mathcal{O}(10^2)$, so the potential is convergent absolutely .

\section{Details of the  coefficients of the potential}\label{app:cof}
By using the multiplication rules of power series \cite{grad},  
\begin{equation}\label{equ:mulrule}
	\sum_{k=0}^{\infty} a_k x^k   \sum_{k=0}^{\infty} b_k x^k = \sum_{k=0}^{\infty} c_k x^k   \,, \quad c_n = \sum_{k=0}^{n} a_k b_{n-k}\,,       
\end{equation}
one can obtain the following relations between the coefficients of $a_n, b_n, c_n$ and $d_n$ from Eqs.(\ref{equ:reca}), (\ref{equ:recb}) -(\ref{equ:recd})
\begin{eqnarray}
	\sum_{n=1}^{\infty}  2n a_n \phi^{2n} \sum_{n=1}^{\infty}  2n a_n \phi^{2n} &=&   \sum_{n=1}^{\infty}  a_n \phi^{2n} \sum_{n=1}^{\infty}  a_n \phi^{2n}\sum_{n=0}^\infty b_n \phi^{2n} \,, \label{equ:use1}\\
	\sum_{n=1}^{\infty} 2n(2n-1)a_n \phi^{2n}  &=&  \sum_{n=1}^{\infty} a_n \phi^{2n} \sum_{n=0}^{\infty} c_n \phi^{2n} \,, \label{equ:use2}\\
	\sum_{n=1}^{\infty} 2n a_n \phi^{2n}  \sum_{n=1}^{\infty} 4n(n+1)(2n+1) a_{n+1} \phi^{2n}  &=&  \sum_{n=1}^{\infty} a_n \phi^{2n} \sum_{n=1}^{\infty} a_n \phi^{2n}  \sum_{n=0}^{\infty} d_n \phi^{2n} \,.\label{equ:use3}
\end{eqnarray}
Multiply  Eq.(\ref{equ:use2}) by $ 4\sum_{n=1}^{\infty} a_n \phi^{2n}$, Eq.(\ref{equ:use3}) by $4A$,  then subtract the summation of them from Eq.(\ref{equ:use1}) multiplied by $(2-B)$, we finally obtain
\begin{eqnarray}
\nonumber
	&& \sum_{n=1}^{\infty}  a_n \phi^{2n} \sum_{n=1}^{\infty}  a_n \phi^{2n}\sum_{n=0}^\infty \bigg[ (2-B)b_n - 4c_n - 4A d_n \bigg] \phi^{2n}   \\
	 &=& \sum_{n=1}^{\infty}  2n a_n \phi^{2n} \sum_{n=1}^{\infty}   \bigg[ [  2(2-B)n+4 ] a_n 
	 -16An(n+1)(2n+1) a_{n+1}  \bigg] \phi^{2n}
	 - 16\sum_{n=1}^{\infty} a_n \phi^{2n} \sum_{n=1}^{\infty}n^2a_n \phi^{2n}   \,.
\end{eqnarray}
Again, by using Eq.(\ref{equ:mulrule}) and the expression of $a_n$ from Eq.(\ref{equ:reca}), one can see that the right hand side of the above equation is indeed vanished. Thus, we obtain the relation (\ref{equ:bcd}).

Leading terms of the coefficients of $a_n, b_n, c_n, d_n$ and $f_n$ with $B=0$ from Eqs.(\ref{equ:reca}), (\ref{equ:recb})-(\ref{equ:recd}) and (\ref{equ:recf}). Here $a_1$ is set to be unity, and $A$ is a parameter in the potential. In the limit of $A\rightarrow \infty$, the only non vanished values of the coefficients are $a_1 =1, b_0=4$ and $c_0=2$, which are just the same as the results in the chaotic inflation with a quadratic potential . The relation $(4-B)b_c - 4c_n = 4A d_n$ is  satisfied for all the values of $n=0,1,\cdots $.

\begin{eqnarray}
a_1 &=& +1 \,,\\
a_2 &=& -\frac{B}{48 A}\,,\\
a_3 &=&+\frac{B (B+3)}{5760 A^2}\,,\\
a_4 &=&-\frac{B \left(B^2+3 B+20\right)}{1290240 A^3}\,,\\
a_5 &=&+\frac{B \left(2 B^3+27 B^2-177 B+420\right)}{928972800 A^4}\,,\\
a_6 &=&-\frac{B \left(2 B^4-93 B^3+2319 B^2-8868 B+6048\right)}{490497638400 A^5}\,,\\
a_7 &=&+\frac{B \left(2 B^5+1293 B^4-33530 B^3+212991 B^2-353652 B+110880\right)}{357082280755200 A^6}\,,\\
a_8 &=&-\frac{B \left(2 B^6-20547 B^5+665142 B^4-5886141 B^3+16998072 B^2-14707248 B+2471040\right)}{342798989524992000 A^7}\,,\\
\nonumber
\cdots & &
\end{eqnarray}

\begin{eqnarray}
b_0 &=& 4\,, \\
b_1 &=&-\frac{B}{6 A}\,, \\
b_2 &=&\frac{B (B+8)}{960 A^2}\,, \\
b_3 &=&\frac{B \left(B^2+3 B-36\right)}{96768 A^3}\,, \\
b_4 &=&\frac{B \left(B^3-14 B^2-96 B+160\right)}{11059200 A^4}\,, \\
b_5 &=&\frac{B \left(B^4-52 B^3-12 B^2+1000 B-672\right)}{1362493440 A^5}\,, \\
b_6 &=&\frac{B \left(1382 B^5-178062 B^4+1122219 B^3+2861064 B^2-10455120 B+3628800\right)}{243465191424000 A^6}\,, \\
b_7 &=&\frac{B \left(4 B^6-1133 B^5+16482 B^4-24005 B^3-126888 B^2+190224 B-38016\right)}{94175546572800 A^7}\,, \\
\nonumber
\cdots & &
\end{eqnarray}

\begin{eqnarray}
c_0 &=& 2 \,,\\
c_1 &=&-\frac{5 B}{24 A}\,,\\
c_2 &=&\frac{B (B+28)}{1920 A^2}\,,\\
c_3 &=&\frac{B \left(2 B^2+111 B-324\right)}{387072 A^3}\,,\\
c_4 &=&\frac{B \left(B^3+116 B^2-996 B+880\right)}{22118400 A^4}\,,\\
c_5 &=&\frac{B \left(2 B^4+545 B^3-8472 B^2+21448 B-8736\right)}{5449973760 A^5}\,,\\
c_6 &=&\frac{B \left(1382 B^5+987648 B^4-22112811 B^3+104548104 B^2-123935760 B+27216000\right)}{486930382848000 A^6}\,,\\
c_7 &=&\frac{B \left(8 B^6+16231 B^5-466176 B^4+3372995 B^3-7622088 B^2+5066640 B-646272\right)}{376702186291200 A^7}\,,\\
\nonumber
\cdots & &
\end{eqnarray}

\begin{eqnarray}
d_0 &=& -\frac{B}{A} \,,\\
d_1 &=&\frac{B (B+3)}{24 A^2}\,,\\
d_2 &=&-\frac{B \left(B^2+8 B+40\right)}{3840 A^3}\,,\\
d_3 &=&-\frac{B \left(B^3+3 B^2+69 B-252\right)}{387072 A^4}\,,\\
d_4 &=&-\frac{B \left(B^4-14 B^3+164 B^2-1640 B+1440\right)}{44236800 A^5}\,,\\
d_5 &=&-\frac{B \left(B^5-52 B^4+637 B^3-7448 B^2+18776 B-7392\right)}{5449973760 A^6}\,,\\
d_6 &=&-\frac{B \left(1382 B^6-178062 B^5+3453639 B^4-43608996 B^3+192918960 B^2-223332480 B+47174400\right)}{973860765696000 A^7}\,,\\
d_7 &=&\frac{B \left(-4 B^7+1133 B^6-34979 B^5+523145 B^4-3294117 B^3+7178088 B^2-4648176 B+570240\right)}{376702186291200 A^8}\,,\\
\nonumber
\cdots & &
\end{eqnarray}

\begin{eqnarray}
f_0 &=& 1 \,,\\
f_1 &=&\frac{B}{48 A}\,,\\
f_2 &=&\frac{(B-2) B}{1920 A^2}\,,\\
f_3 &=&\frac{B \left(17 B^2-89 B+60\right)}{1290240 A^3}\,,\\
f_4 &=&\frac{B \left(31 B^3-285 B^2+522 B-168\right)}{92897280 A^4}\,,\\
f_5 &=&\frac{B \left(691 B^4-9389 B^3+31374 B^2-28620 B+5040\right)}{81749606400 A^5}\,,\\
f_6 &=&\frac{B \left(5461 B^5-99733 B^4+509592 B^3-874296 B^2+462072 B-47520\right)}{25505877196800 A^6}\,,\\
\nonumber
f_7 &=&\frac{B \left(1859138 B^6-43007343 B^5+304039386 B^4-817449429 B^3+838021080 B^2-278999280 B+17297280\right)}{342798989524992000 A^7}\,,\\
\cdots & &
\end{eqnarray}

\end{document}